\documentclass[11pt]{article}%
\pdfoutput=1
\usepackage{amsfonts}
\usepackage{amsmath}
\usepackage{amssymb}
\usepackage{graphicx}
\usepackage{fullpage}%
\providecommand{\U}[1]{\protect\rule{.1in}{.1in}}
\newtheorem{theorem}{Theorem}

\newtheorem{corollary}[theorem]{Corollary}

\newtheorem{definition}[theorem]{Definition}

\newtheorem{fact}[theorem]{Fact}
\newtheorem{lemma}[theorem]{Lemma}

\newtheorem{remark}[theorem]{Remark}

\newenvironment{proof}[1][Proof]{\noindent\textbf{#1.} }{\ \rule{0.5em}{0.5em}}

\let\originalleft\left
\let\originalright\right
\def\left#1{\mathopen{}\originalleft#1}
\def\right#1{\originalright#1\mathclose{}}

\begin{document}

\title{\textbf{Recursive quantum convolutional encoders are catastrophic: A simple
proof}}
\author{Monireh Houshmand\\\textit{\small Electrical Engineering Department,}
\\\textit{\small Imam Reza International University,}\\\textit{{\small {Mashhad, Iran}}} \and
Mark M.~Wilde\\\textit{{\small {School of Computer Science, McGill University,}}}\\\textit{{\small {Montreal, Quebec H3A 2A7, Canada}}}}
\date{\today}
\maketitle

\begin{abstract}
Poulin, Tillich, and Ollivier discovered an important separation between the
classical and quantum theories of convolutional coding, by proving that a
quantum convolutional encoder cannot be both non-catastrophic and recursive.
Non-catastrophicity is desirable so that an iterative decoding algorithm
converges when decoding a quantum turbo code whose constituents are quantum
convolutional codes, and recursiveness is as well so that a quantum turbo code
has a minimum distance growing nearly linearly with the length of the code,
respectively. Their proof of the aforementioned theorem was admittedly
\textquotedblleft rather involved,\textquotedblright\ and as such, it has been
desirable since their result to find a simpler proof. In this paper, we
furnish a proof that is arguably simpler. Our approach is group-theoretic---we
show that the subgroup of memory states that are part of a zero
physical-weight cycle of a quantum convolutional encoder is equivalent to the
centralizer of its \textquotedblleft finite-memory\textquotedblright\ subgroup
(the subgroup of memory states which eventually reach the identity memory
state by identity operator inputs for the information qubits and identity or
Pauli-$Z$ operator inputs for the ancilla qubits). After proving that this
symmetry holds for any quantum convolutional encoder, it easily follows that
an encoder is non-recursive if it is non-catastrophic. Our proof also
illuminates why this no-go theorem does not apply to entanglement-assisted
quantum convolutional encoders---the introduction of shared entanglement as a
resource allows the above symmetry to be broken.

\end{abstract}

\section{Introduction}

Quantum convolutional coding is an approach to quantum error correction that
allows a sender to encode a stream of quantum data in an online fashion, by
the repeated application of some encoding unitary \cite{PhysRevLett.91.177902}%
. As soon as the qubits have been encoded, the sender can then transmit them
over a noisy quantum channel to a receiver, whereupon the receiver can begin
to decode the qubits as he receives them. An advantage of this approach is
that the encoding complexity is linear in the number of qubits that are
encoded, and the most clear application of these codes is in the context of
quantum communication. One can also realize a quantum turbo code by serially
concatenating two quantum convolutional encoders \cite{PTO09}. Quantum turbo
codes have linear encoding and decoding complexity, and they are known (from
numerical simulations) to exhibit good performance when the noisy channel is a
depolarizing channel \cite{PTO09,WHB13}.

The notions of a state diagram, catastrophicity, and recursiveness are
prominent in the classical theory of convolutional coding
\cite{V71,book1999conv,McE02}, and Poulin, Tillich, and Ollivier (PTO) were
the first to provide satisfying definitions of these notions for a quantum
convolutional encoder \cite{PTO09}. In fact, a casual glance at the
development in Ref.~\cite{PTO09} might lead one to believe that there is
little difference between the classical and quantum theories of convolutional
coding, given that the classical definitions of state diagram,
catastrophicity, and recursiveness were essentially imported\ and then adapted
to the quantum setting.

In spite of the apparent similarities, PTO found a striking separation between
the classical and quantum theories of convolutional coding:\ in the quantum
case, encoders cannot be both non-catastrophic and recursive \cite{PTO09}.
Recall that in the classical case, convolutional encoders can certainly
possess both properties, and in fact, one looks for encoders satisfying both
when constructing classical turbo codes \cite{BDMP98}. Non-catastrophicity
ensures that an iterative decoding algorithm for a turbo code converges, while
recursiveness of one of the convolutional encoders ensures that the turbo code
has a minimum distance growing nearly linearly with the number of qubits being
encoded (this is true for both the classical and quantum cases
\cite{KU98,PTO09,AT11}). The important theorem of PTO\ was then later extended
to apply to quantum convolutional encoders for subsystem codes \cite{WH10b},
but Ref.~\cite{WH10b} also pointed out that this\ theorem does not apply if a
sender and receiver share entangled Bell states before communication begins
(in fact, the theory of entanglement-assisted quantum convolutional coding
bears more similarities with the classical theory than does the
\textquotedblleft unassisted\textquotedblright\ theory of quantum
convolutional coding).

The purpose of the present paper is to simplify the proof of the PTO\ theorem,
stated as \textquotedblleft recursive quantum convolutional encoders are
catastrophic.\textquotedblright\ One might deem it unnecessary to do so given
that PTO\ have already provided a proof, but it is often the case that if a
theorem is sufficiently important, then alternate proofs can provide more
insight into the theorem itself, perhaps spark new developments, or be useful
for pedagogical purposes. Since this theorem gives such a strong separation
between the classical and quantum theories, it clearly stands as one of the
most important results in the theory of quantum convolutional coding. In the
case of the PTO\ theorem, a quick glance over its original proof indicates
that it is complicated. Indeed, Ref.~\cite{PTO09}\ states that
\textquotedblleft the proof of Theorem 1 is rather involved,\textquotedblright%
\ and as such, they leave its many details to an appendix.

We now provide a brief summary of our proof of the PTO\ theorem for the expert
in quantum convolutional coding (this summary should become more
understandable for the non-expert after reviewing the background material in
the next section). First, it is apparently simpler to prove the
logically-equivalent contrapositive of the PTO\ theorem. That is, we prove the
statement \textquotedblleft If a quantum convolutional encoder is
non-catastrophic, then it is non-recursive.\textquotedblright\ Our approach
relies heavily on the use of group theory. We define a finite-memory subgroup
$\mathcal{F}_{0}$, which consists of all the memory states of a quantum
convolutional encoder that eventually reach the identity memory state via a
path in the state diagram resulting from applying combinations of identity and
Pauli-$Z$ operators on the ancilla qubits and identity operators on the
information qubits (we call such paths \textquotedblleft finite standard
paths\textquotedblright).\footnote{This formulation of the finite-memory
subgroup is different from that in Ref.~\cite{PTO09}, and it is a further step
that is helpful in simplifying the proof of the PTO\ theorem.} We define the
infinite-memory set to consist of all memory states that do not have this
property. This set on its own is not a subgroup, but if we quotient the full
Pauli group acting on the memory qubits by the finite-memory subgroup, then we
obtain a subgroup of infinite-memory states that obey the group property of
staying within the infinite-memory set when they are multiplied by each other.
We let $\mathcal{I}_{0}$ denote this subgroup, and for simplicity, we refer to
it as the \textquotedblleft infinite-memory subgroup.\textquotedblright\ We
then show that all memory states in the finite-memory subgroup $\mathcal{F}%
_{0}$ commute with the memory states that are part of a zero physical-weight
cycle (these latter states form a subgroup that we denote by $\mathcal{P}_{0}%
$). After doing so, we prove the converse statement, namely, that any memory
state which commutes with all elements of the finite-memory subgroup is part
of some zero-physical weight cycle. This establishes that the zero
physical-weight cycle subgroup $\mathcal{P}_{0}$ is equivalent to the
centralizer $C\left(  \mathcal{F}_{0}\right)  $ of the finite-memory subgroup
$\mathcal{F}_{0}$. This last observation easily leads to the statement of the
theorem. For a non-catastrophic encoder, any weight-one logical input drives
it to a memory state that commutes with all the elements of $C\left(
\mathcal{F}_{0}\right)  $, implying that this memory state belongs to
$\mathcal{F}_{0}$ (here we are invoking a double centralizer theorem that
applies for the Pauli group). Since we can conclude that this memory state
belongs to $\mathcal{F}_{0}$, it follows that there is a finite standard path
from it to the identity memory state, and as such, a non-catastrophic encoder
is non-recursive.

We structure this paper as follows. In the next section, we review the
definition of a quantum convolutional code, a quantum convolutional encoder, a
state diagram, catastrophicity, and recursiveness. Ref.~\cite{PTO09}%
\ established these definitions, but we repeat them here for convenience and
invite the expert to skip the next section. Section~\ref{sec:main-theorem}%
\ presents our proof of the PTO\ theorem along the lines outlined in the
previous paragraph. Section~\ref{sec:discussion}\ provides a brief discussion
to conclude this paper.

\section{Background}

\subsection{Quantum convolutional codes}

We begin by recalling some standard facts, and then we review the definition
of a quantum convolutional code. A Pauli sequence is a countably-infinite
tensor product of Pauli matrices:%
\[
\mathbf{A}=\bigotimes\limits_{i=0}^{\infty}A_{i},
\]
where each operator $A_{i}$ in the sequence is an element of the Pauli
group~$\Pi\equiv\left\{  I,X,Y,Z\right\}  $.\footnote{In our work, to be
precise, we are dealing with the quotient of the Pauli group by its center, so
that global phases are irrelevant. For simplicity, in this paper we refer to the quotient of the
Pauli group by its center as ``the Pauli group.''} Let $\Pi^{\mathbb{Z}^{+}}$ denote the set
of all Pauli sequences. A Pauli sequence is finite-weight if only finitely
many operators~$A_{i}$ in the sequence are equal to $X$, $Y$, or $Z$, and it
is an infinite-weight sequence otherwise.

\begin{definition}
[Quantum Convolutional Code]\label{def:QCC}A rate-$k/n$ quantum convolutional
code admits a representation with a basic set $\mathcal{G}_{0}$ of $n-k$
generators and all of their $n$-qubit shifts:%
\[
\mathcal{G}_{0}\equiv\left\{  \mathbf{G}_{i}\in\Pi^{\mathbb{Z}^{+}}:1\leq
i\leq n-k\right\}  .
\]
In order to form a quantum convolutional code, these generators should commute
with themselves and all of the $n$-qubit shifts of themselves and the other generators.
\end{definition}

Equivalently, a rate-$k/n$ quantum convolutional code is specified by $n-k$
generators $h_{1}$, $h_{2}$, $\ldots$, $h_{n-k}$, where%
\begin{equation}%
\begin{array}
[c]{cc}%
h_{1} & =\\
h_{2} & =\\
\vdots & \\
h_{n-k} & =
\end{array}%
\begin{array}
[c]{c}%
h_{1,1}\\
h_{2,1}\\
\vdots\\
h_{n-k,1}%
\end{array}%
\begin{array}
[c]{c}%
|\\
|\\
\vdots\\
|
\end{array}%
\begin{array}
[c]{c}%
h_{1,2}\\
h_{2,2}\\
\vdots\\
h_{n-k,2}%
\end{array}%
\begin{array}
[c]{c}%
|\\
|\\
\vdots\\
|
\end{array}%
\begin{array}
[c]{c}%
\cdots\\
\cdots\\
\ \\
\cdots
\end{array}%
\begin{array}
[c]{c}%
|\\
|\\
\vdots\\
|
\end{array}%
\begin{array}
[c]{c}%
h_{1,d_{1}}\\
h_{2,d_{2}}\\
\vdots\\
h_{n-k,d_{n-k}}%
\end{array}
. \label{eq:general-stab-generators}%
\end{equation}
Each entry $h_{i,j}$ is an $n$-qubit Pauli operator and $d_{i}$ is the degree
of generator$~h_{i}$ (in general, the degrees $d_{i}$ can be different from
each other). We obtain the other generators of the code by shifting the above
generators to the right by multiples of $n$ qubits. (In the above, note that
the entries $h_{1,d_{1}}$, $h_{2,d_{2}}$, \ldots, $h_{n-k,d_{n-k}}$ are not
required to be in the same column, but we have written them in the above way
for convenience.) Note also that the generators could have an infinite degree
$d_{i}$, but in this case, they should have some periodicity. We can obtain
infinite-weight generators from a finite-weight representation, for example,
by multiplying the first generator by the second one and all of its shifts,
and this gives a different representation of the same code.

\subsection{Quantum convolutional encoders}%

\begin{figure}
[ptb]
\begin{center}
\includegraphics[
natheight=4.333600in,
natwidth=5.586700in,
height=3.2007in,
width=3.2396in
]%
{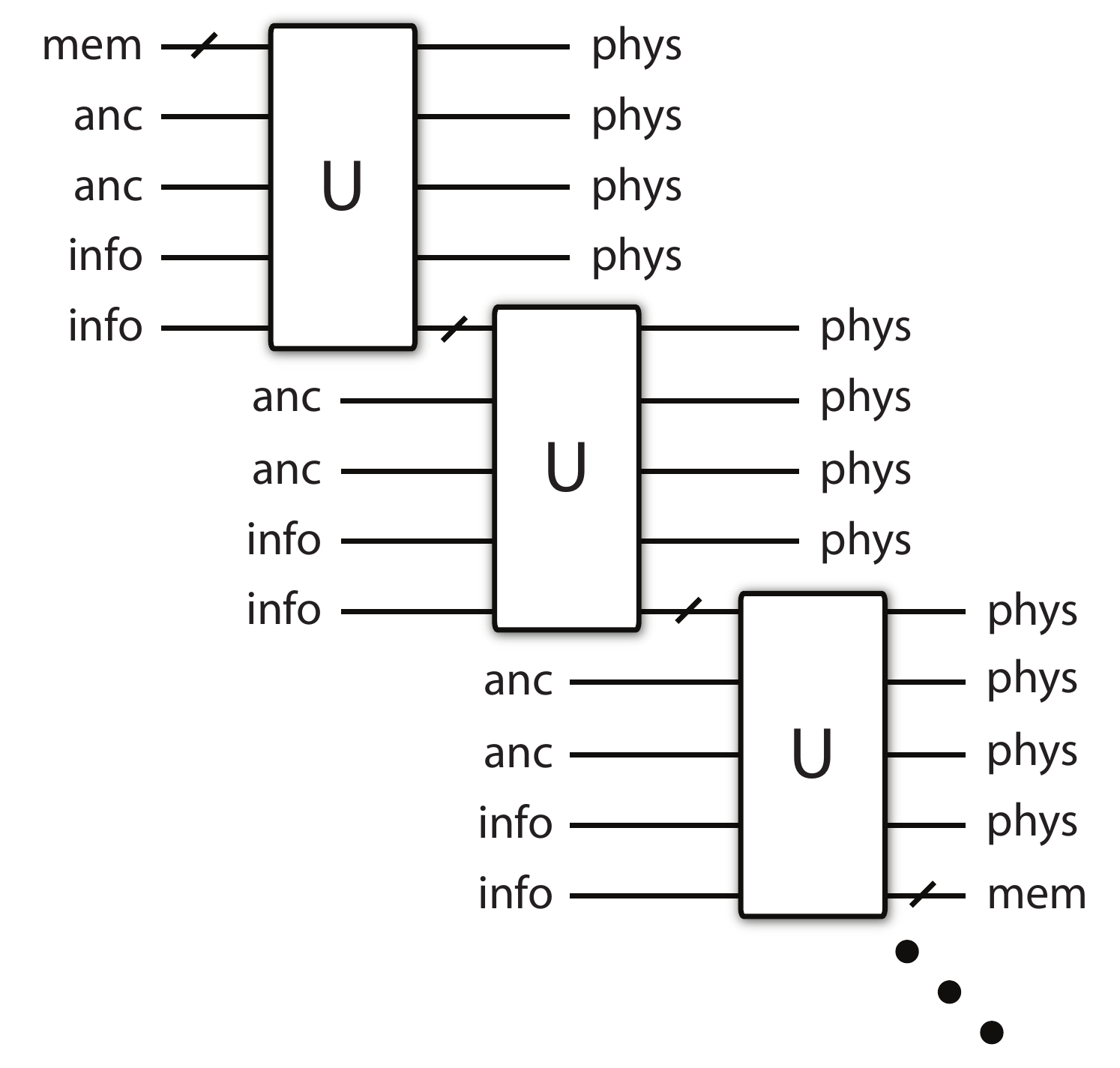}%
\caption{The encoder $U$ for a quantum convolutional code that has four
physical qubits for every two information qubits. The encoder $U$ acts on $m$
memory qubits, two ancilla qubits, and two information qubits to produce four
output physical qubits to be sent over the channel and $m$ output memory
qubits to be fed into the next round of encoding.}%
\label{fig:conv-encoder}%
\end{center}
\end{figure}

Figure~\ref{fig:conv-encoder}\ depicts an example of an encoder for a quantum
convolutional code. The encoder depicted there has four physical qubits for
every two information qubits. The unencoded qubit stream might have the
following form:%
\begin{equation}
\left\vert 0\right\rangle \left\vert 0\right\rangle \left\vert \phi
_{1}\right\rangle \left\vert \phi_{2}\right\rangle \left\vert 0\right\rangle
\left\vert 0\right\rangle \left\vert \phi_{3}\right\rangle \left\vert \phi
_{4}\right\rangle \cdots, \label{eq:unencoded-stream}%
\end{equation}
so that an ancilla qubit appears as every first and second qubit and an
information qubit appears as every third and fourth qubit (generally, these
information qubits can be entangled with each other and even with an
inaccessible reference system, but we write them as product states for simplicity).

More generally, a convolutional encoder acts on some number$~m$ of memory
qubits, $n-k$ ancilla qubits, and $k$~information qubits, and it produces
$n$~output physical qubits and $m$~output memory qubits to be fed into the
next round of encoding. It should transform an unencoded Pauli $Z$ operator
acting on the $i^{\text{th}}$ ancilla qubit to the $i^{\text{th}}$ stabilizer
generator $h_{i}$ in (\ref{eq:general-stab-generators}). That is, it should be
some Clifford transformation.\footnote{A Clifford transformation is a unitary
operator that preserves the Pauli group under unitary conjugation.} The first
application of the encoder $U$ results in some Pauli operator $g_{i,1}$ acting
on the $m$ output memory qubits. The second application of the encoder $U$
results in some other Pauli operator $g_{i,2}$ acting on the $m$ output memory
qubits and so on. The shift invariance of the overall encoding guarantees that
shifts of the unencoded$~Z$ Pauli operators transform to appropriate shifts of
the generators. A convolutional encoder for the code  performs the following
transformation:%
\begin{equation}%
\begin{tabular}
[c]{c|c|c}%
Mem. & \multicolumn{1}{|c|}{Anc.} & Info.\\\hline\hline
$I^{\otimes m}$ & $Z(1)$ & $I^{\otimes k}$\\
$g_{1,1}$ & \multicolumn{1}{|c|}{$I^{\otimes n-k}$} & $I^{\otimes k}$\\
$\vdots$ & \multicolumn{1}{|c|}{$\vdots$} & $\vdots$\\
$g_{1,d_{1}-2}$ & \multicolumn{1}{|c|}{$I^{\otimes n-k}$} & $I^{\otimes k}$\\
$g_{1,d_{1}-1}$ & \multicolumn{1}{|c|}{$I^{\otimes n-k}$} & $I^{\otimes k}%
$\\\hline
$I^{\otimes m}$ & $Z(2)$ & $I^{\otimes k}$\\
$g_{2,1}$ & \multicolumn{1}{|c|}{$I^{\otimes n-k}$} & $I^{\otimes k}$\\
$\vdots$ & \multicolumn{1}{|c|}{$\vdots$} & $\vdots$\\
$g_{2,d_{2}-2}$ & \multicolumn{1}{|c|}{$I^{\otimes n-k}$} & $I^{\otimes k}$\\
$g_{2,d_{2}-1}$ & \multicolumn{1}{|c|}{$I^{\otimes n-k}$} & $I^{\otimes k}%
$\\\hline
$\vdots$ & \multicolumn{1}{|c|}{$\vdots$} & $\vdots$\\\hline
$I^{\otimes m}$ & $Z(s)$ & $I^{\otimes k}$\\
$g_{s,1}$ & \multicolumn{1}{|c|}{$I^{\otimes n-k}$} & $I^{\otimes k}$\\
$\vdots$ & \multicolumn{1}{|c|}{$\vdots$} & $\vdots$\\
$g_{s,d_{s}-2}$ & \multicolumn{1}{|c|}{$I^{\otimes n-k}$} & $I^{\otimes k}$\\
$g_{s,d_{s}-1}$ & \multicolumn{1}{|c|}{$I^{\otimes n-k}$} & $I^{\otimes k}$%
\end{tabular}
\rightarrow%
\begin{tabular}
[c]{c|c}%
Phys. & Mem.\\\hline\hline
$h_{1,1}$ & $g_{1,1}$\\
$h_{1,2}$ & $g_{1,2}$\\
$\vdots$ & $\vdots$\\
$h_{1,d_{1}-1}$ & $g_{1,d_{1}-1}$\\
$h_{1,d_{1}}$ & $I^{\otimes m}$\\\hline
$h_{2,1}$ & $g_{2,1}$\\
$h_{2,2}$ & $g_{2,2}$\\
$\vdots$ & $\vdots$\\
$h_{2,d_{2}-1}$ & $g_{2,d_{2}-1}$\\
$h_{2,d_{2}}$ & $I^{\otimes m}$\\\hline
$\vdots$ & $\vdots$\\\hline
$h_{s,1}$ & $g_{s,1}$\\
$h_{s,2}$ & $g_{s,2}$\\
$\vdots$ & $\vdots$\\
$h_{s,d_{s}-1}$ & $g_{s,d_{s}-1}$\\
$h_{s,d_{s}}$ & $I^{\otimes m}$%
\end{tabular}
\label{eq:general-encoder}%
\end{equation}
where, as a visual aid, we have separated the Pauli operators acting on the
memory qubits, ancilla qubits, and information qubits at the input with a
vertical bar and we have done the same for those acting on the physical qubits
and memory qubits at the output. The parameter $m$ is the number of memory
qubits, $k$ is the number of information qubits, $n-k$ is the number of
ancilla qubits, and we make the abbreviation $s\equiv n-k$. A horizontal bar
separates the entries of the encoder needed to encode the first generator from
the entries needed to encode the second generator. Each $g_{i,j}$ is a Pauli
operator acting on some number~$m$ of memory qubits---these operators should
be consistent with the input-output commutation relations of the encoder
\cite{arXiv:1105.0649v1}. We stress that the above input-output relations only
partially specify the encoder such that it produces a code with the stabilizer
generators in (\ref{eq:general-stab-generators}), and there is still a fair
amount of freedom remaining in the encoding. Note that it could be the case
that the encoder transforms a Pauli-$Z$ operator acting on an ancilla qubit to
a Pauli operator with infinite weight (though, the resulting operator would
have some periodicity because the encoder acts on a finite number of qubits).

\subsection{State diagram}

The state diagram for a quantum convolutional encoder is the most important
tool for determining whether it is catastrophic or recursive~\cite{PTO09}. The
state diagram is similar to one for a classical
encoder~\cite{V71,book1999conv,McE02}, with an important exception for the
quantum case that incorporates the fact that the logical operators of a
quantum code are unique up to multiplication by the stabilizer generators. The
state diagram allows us to analyze the \textquotedblleft
flow\textquotedblright\ of the logical operators through the quantum
convolutional encoder.

\begin{definition}
[State Diagram]The state diagram for a quantum convolutional encoder is a
directed multigraph with $4^{m}$ vertices that we can think of as
\textquotedblleft memory states,\textquotedblright\ where $m$ is the number of
memory qubits in the encoder. Each memory state corresponds to an $m$-qubit
Pauli operator~$M$ that acts on the memory qubits. We connect two vertices $M$
and $M^{\prime}$ with a directed edge from $M$ to $M^{\prime}$ and label this
edge as $\left(  L,P\right)  $ if the encoder takes the $m$-qubit Pauli
operator~$M$, an $\left(  n-k\right)  $-qubit Pauli operator $S\in\left\{
I,Z\right\}  ^{n-k}$ acting on the $n-k$ ancilla qubits, and a $k$-qubit Pauli
operator $L$ acting on the information qubits, to an $n$-qubit Pauli operator
$P$ acting on the $n$ physical qubits and an $m$-qubit Pauli
operator~$M^{\prime}$ acting on the $m$ memory qubits:%
\[%
\begin{tabular}
[c]{c|c|c}%
\emph{Mem.} & \emph{Anc.} & \emph{Info.}\\\hline\hline
$M$ & $S$ & $L$%
\end{tabular}
\ \ \ \ \underrightarrow{\ \text{encoder\ }}\ \ \
\begin{tabular}
[c]{c|c}%
\emph{Phys.} & \emph{Mem.}\\\hline\hline
$P$ & $M^{\prime}$%
\end{tabular}
\ \ \ .
\]
The labels $L$ and $P$ are the respective logical and physical labels of the edge.
\end{definition}

\subsection{Catastrophicity}

We now review the definition of catastrophicity from Ref.~\cite{PTO09}, which
is based on the classical notion of catastrophicity from
Refs.~\cite{V71,McE02}. The essential idea behind catastrophic error
propagation is that an error with finite weight, after being fed through the
inverse of the encoder, could propagate infinitely throughout the decoded
information qubit stream without triggering syndromes corresponding to these
errors. The only way that this catastrophic error propagation can occur is if
there is some cycle in the state diagram where all of the edges along the
cycle have physical labels equal to the identity operator, while at least one
of the edges has a logical label that is not equal to the identity. If such a
cycle exists, it implies that the finite-weight channel error produces an
infinite-weight information qubit error without triggering syndrome bits
corresponding to this error (if it did trigger syndrome bits, this cycle would
not be in the state diagram), and an iterative decoding algorithm such as that
presented in Ref.~\cite{PTO09}\ is not able to detect these errors. So, we can
now state the definition of a catastrophic encoder.

\begin{definition}
[Catastrophic Encoder]\label{def:catastrophic}A quantum convolutional encoder
acting on memory qubits, information qubits, and ancilla qubits is
catastrophic if there exists a cycle in its state diagram where all edges in
the cycle have zero physical weight, but there is at least one edge in the
cycle with non-zero logical weight.
\end{definition}

\subsection{Recursiveness}

Recursiveness or lack thereof is a fundamental property of a quantum
convolutional encoder as discussed in Refs.~\cite{PTO09,AT11}.

\begin{definition}
[Path]A path in the state diagram is a sequence of memory states $M_{1}$,
$M_{2}$, \ldots\ (with corresponding edges) such that $M_{i}\rightarrow
M_{i+1}$ belongs to the state diagram.
\end{definition}

\begin{definition}
[Recursive encoder]\label{def:recursive}An admissable path is a path in the
state diagram for which its first edge is not part of a zero physical-weight
cycle. Consider any vertex belonging to a zero physical-weight loop and any
admissable path beginning at this vertex that also has logical weight one. The
encoder is recursive if all such paths do not contain a zero physical-weight loop.
\end{definition}

We can gain some intuition behind the above definition by recalling the
definition of a recursive classical convolutional encoder. In the classical
case, an encoder is recursive if it has an infinite impulse response---that
is, if it outputs an infinite-weight, periodic sequence in response to an
input consisting of a single \textquotedblleft one\textquotedblright\ followed
by an infinite number of \textquotedblleft zeros.\textquotedblright%
\ Definition~\ref{def:recursive} above for the quantum case ensures that the
response to a single non-identity Pauli operator (one of $\{X,Y,Z\}$) at a
single logical input along with the identity operator at all other logical
inputs leads to a periodic output sequence of Pauli operators with infinite
weight. Though, the definition above ensures that this is not only the case
for the above sequence but also for one in which the ancilla qubit inputs can
be chosen arbitrarily from $\left\{  I,Z\right\}  $. Thus, it is a much more
stringent condition for a quantum convolutional encoder to be recursive.

\section{Main result:\ A simple proof that recursive quantum convolutional
encoders are catastrophic}

\label{sec:main-theorem}We now provide a proof of the PTO\ theorem. Our
strategy is to show that the subgroup consisting of memory states in zero
physical-weight cycles is equivalent to the centralizer of the finite-memory
subgroup (which we define below). After doing so, it follows rather easily
that a non-catastrophic encoder must be non-recursive.

A zero physical-weight cycle has the following structure:%
\begin{equation}%
\begin{tabular}
[c]{c|c|c}%
Mem. & Anc. & Info.\\\hline\hline
$M_{1}$ & $S_{1}$ & $L_{1}$\\
$M_{2}$ & $S_{2}$ & $L_{2}$\\
$\vdots$ & $\vdots$ & $\vdots$\\
$M_{p}$ & $S_{p}$ & $L_{p}$%
\end{tabular}
\ \ \rightarrow%
\begin{tabular}
[c]{c|c}%
Phys. & Mem.\\\hline\hline
$I^{\otimes n}$ & $M_{2}$\\
$I^{\otimes n}$ & $M_{3}$\\
$\vdots$ & $\vdots$\\
$I^{\otimes n}$ & $M_{1}$%
\end{tabular}
\ \ \ \ , \label{eq:physical-weight-cycle}%
\end{equation}
where $M_{1}$, \ldots, $M_{p}$ are arbitrary $m$-qubit Pauli operators acting
on the $m$ memory qubits, the operators $S_{i}\in\{I,Z\}^{\otimes(n-k)}$ act
on the $n-k$ ancilla qubits, and the operators $L_{i}$ are arbitrary $k$-qubit
Pauli operators acting on the $k$ information qubits. A zero physical-weight
cycle with non-zero logical weight is a zero physical-weight cycle such that
at least one of the operators $L_{i}$ is not equal to the identity operator.

Let $\mathcal{P}_{0}$ denote the subgroup of the Pauli group on $m$ qubits
consisting of memory states that are part of a zero physical-weight cycle.
That $\mathcal{P}_{0}$ is a subgroup follows from the observation that if
$M_{1}$ is part of a zero physical-weight cycle and if $M_{2}$ is part of a
zero physical-weight cycle, then $M_{1}\cdot M_{2}$ must be as well because we
can combine the cycles of which $M_{1}$ and $M_{2}$ are a part simply by
multiplying each transition of the individual cycles to get the new
transitions for the zero physical-weight cycle of which $M_{1}\cdot M_{2}$ is involved.

\begin{definition}
[Standard path]Consider a given memory state $M_{0}$. A standard path
originating from $M_{0}$ is a path with the following structure:%
\begin{equation}%
\begin{tabular}
[c]{c|c|c}%
Mem. & Anc. & Info.\\\hline\hline
$M_{0}$ & $S_{0}$ & $I^{\otimes k}$\\
$M_{1}$ & $S_{1}$ & $I^{\otimes k}$\\
$M_{2}$ & $S_{2}$ & $I^{\otimes k}$\\
$\vdots$ & $\vdots$ & $\vdots$%
\end{tabular}
\ \rightarrow%
\begin{tabular}
[c]{c|c}%
Phys. & Mem.\\\hline\hline
$P_{1}$ & $M_{1}$\\
$P_{2}$ & $M_{2}$\\
$P_{3}$ & $M_{3}$\\
$\vdots$ & $\vdots$%
\end{tabular}
\ \ \ , \label{eq:standard-path}%
\end{equation}
where $S_{i}\in\left\{  I,Z\right\}  ^{\otimes\left(  n-k\right)  }$ and
$P_{i}\in\left\{  I,X,Y,Z\right\}  ^{\otimes n}$. From the above, we can see
that a standard path originating from $M_{0}$ is such that the output memory
state of the $i^{\text{th}}$ row is the same as the input memory state of the
$(i+1)^{\text{th}}$ row, with the ancilla operators chosen from $\left\{
I,Z\right\}  ^{\otimes\left(  n-k\right)  }$ and the input logical operators
equal to the identity for every transition.
\end{definition}

We now define two types of standard paths: finite standard paths and infinite
standard paths.

\begin{definition}
[Finite standard path]A finite standard path is a standard path which has
\textquotedblleft finite duration,\textquotedblright\ i.e., it ends in the
identity memory state after traversing a finite number of memory states for
some $S_{i}\in\left\{  I,Z\right\}  ^{\otimes\left(  n-k\right)  }$:%
\begin{equation}%
\begin{tabular}
[c]{c|c|c}%
Mem. & Anc. & Info.\\\hline\hline
$M_{0}$ & $S_{0}$ & $I^{\otimes k}$\\
$M_{1}$ & $S_{1}$ & $I^{\otimes k}$\\
$M_{2}$ & $S_{2}$ & $I^{\otimes k}$\\
$\vdots$ & $\vdots$ & $\vdots$\\
$M_{t-2}$ & $S_{t-2}$ & $I^{\otimes k}$\\
$M_{t-1}$ & $S_{t-1}$ & $I^{\otimes k}$%
\end{tabular}
\ \rightarrow%
\begin{tabular}
[c]{c|c}%
Phys. & Mem.\\\hline\hline
$P_{1}$ & $M_{1}$\\
$P_{2}$ & $M_{2}$\\
$P_{3}$ & $M_{3}$\\
$\vdots$ & $\vdots$\\
$P_{t-1}$ & $M_{t-1}$\\
$P_{t}$ & $I^{\otimes m}$%
\end{tabular}
\ \ \ , \label{eq:finite-standard-path}%
\end{equation}
where $M_{0}$, $M_{1}$, \ldots, $M_{t-1}$ are not equal to the identity. A
finite-memory state is a memory state from which a finite standard path
originates. (In (\ref{eq:finite-standard-path}), $M_{0}$, $M_{1}$, \ldots,
$M_{t-1}$ are finite-memory states.)
\end{definition}

\begin{fact}
[Finite-memory subgroup]\label{fact:finite-mem-sub}The finite-memory states
form a subgroup of the Pauli group on $m$ qubits, with the group property
being that each element is the starting memory state for some finite standard
path. That is, if $M_{1}$ and $M_{2}$ are finite-memory states, then the
finite standard path originating from $M_{1}\cdot M_{2}$ is constructed by
multiplying the finite standard paths originating from $M_{1}$ and $M_{2}$
which individually lead to the identity memory state. Thus, this path
eventually ends in the identity memory state and as such is a finite standard
path. Let $\mathcal{F}_{0}$ denote this subgroup (we will refer to it as the
finite-memory subgroup).
\end{fact}

\begin{definition}
[Infinite standard path]An infinite standard path is a standard path which has
infinite duration, i.e., it does not ever end in the identity memory state
when the input on the ancilla qubits is chosen from the set $\left\{
I,Z\right\}  ^{\otimes\left(  n-k\right)  }$ and the input for the logical
qubits is the identity$~I^{\otimes k}$. Since memory state operators act on a
finite number of qubits, the input memory state of the $i^{th}$ row is the
same as the output memory state of the $j^{th}$ row, for some $j\geq i$ and
$S_{0},S_{1},\ldots,S_{i},\ldots,S_{j}\in\left\{  I,Z\right\}  ^{\otimes
\left(  n-k\right)  }$:%
\begin{equation}%
\begin{tabular}
[c]{c|c|c}%
Mem. & Anc. & Info.\\\hline\hline
$M_{0}$ & $S_{0}$ & $I^{\otimes k}$\\
$M_{1}$ & $S_{1}$ & $I^{\otimes k}$\\
$M_{2}$ & $S_{2}$ & $I^{\otimes k}$\\
$\vdots$ & $\vdots$ & $\vdots$\\
$M_{i}$ & $S_{i}$ & $I^{\otimes k}$\\
$\vdots$ & $\vdots$ & $\vdots$\\
$M_{j}$ & $S_{j}$ & $I^{\otimes k}$\\
&  &
\end{tabular}
\ \rightarrow%
\begin{tabular}
[c]{c|c}%
Phys. & Mem.\\\hline\hline
$P_{1}$ & $M_{1}$\\
$P_{2}$ & $M_{2}$\\
$P_{3}$ & $M_{3}$\\
$\vdots$ & $\vdots$\\
$P_{i+1}$ & $M_{i+1}$\\
$\vdots$ & $\vdots$\\
$P_{j+1}$ & $M_{i}$\\
&
\end{tabular}
\ \ \ . \label{eq:infinite-standard-path}%
\end{equation}
In the above, the memory states $M_{0}$, $M_{1}$, \ldots, $M_{j}$ are not
equal to the identity operator. Observe that different sequences $S_{0}$,
\ldots, $S_{j}$ can lead to different cycles starting at some memory state
$M_{i}$, but that for any fixed sequence $S_{0}$, \ldots, $S_{j}$ we always
have a cycle of this form. An infinite-memory state is a memory state which is
never the starting memory state for a finite standard path.
\end{definition}

\begin{fact}
[Infinite-memory set]The set of all infinite-memory states does not form a
subgroup. Indeed, let $M_{1}$ be an infinite-memory state and let $M_{2}$ be a
finite-memory state. The state $M_{3}\equiv M_{1}\cdot M_{2}$ is then an
infinite-memory state that eventually enters the same loop as $M_{1}$.
However, it is clear that the state $M_{3}\cdot M_{1}=M_{2}$ is a
finite-memory state obtained by multiplying two infinite-memory states, so
that all states in the infinite-memory set need not obey the group property.
\end{fact}

\begin{fact}
[Infinite-memory subgroup]If we take the quotient of the full Pauli group
acting on $m$ qubits by the finite-memory subgroup, we obtain a set of
infinite-memory states that do obey the group property. That is, if $M_{1}$
and $M_{2}$ are infinite-memory states in this subgroup, then every standard
path originating from $M_{1}\cdot M_{2}$ is constructed by multiplying every
standard path originating from $M_{1}$ and $M_{2}$. The resulting paths will
never end in the identity memory state no matter which operators in $\left\{
I,Z\right\}  ^{\otimes\left(  n-k\right)  }$ are input for the ancilla qubits.
Let $\mathcal{I}_{0}$ denote this subgroup (for simplicity, we will refer to
it as the infinite-memory subgroup, and our convention is that the identity
memory state is part of this subgroup). Thus, all memory states in this
subgroup are infinite-memory and we obtain only infinite-memory states when
multiplying them together.
\end{fact}

One can easily construct a generating set for the infinite-memory subgroup
explicitly if a generating set $\left\{ T_{1},\ldots,T_{l}\right\} $
for the finite-memory subgroup is available. Indeed, we would just need to
find $2m-l$ independent generators $\left\{ U_{l+1},\ldots,U_{2m}%
\right\} $ such that $\left\{ T_{1},\ldots,T_{l}, U_{l+1}, \ldots,
U_{2m} \right\} $
constitutes a generating set for the Pauli group on the $m$ memory qubits. The
generating set $\left\{ U_{l+1},\ldots,U_{2m}\right\} $ then
generates the infinite-memory subgroup mentioned above.

Observe that the finite-memory subgroup $\mathcal{F}_{0}$ and the
infinite-memory subgroup $\mathcal{I}_{0}$\ are independent subgroups, and
together they constitute the full Pauli group acting on $m$ qubits.

\begin{lemma}
\label{lemma1} The memory states in a zero physical weight cycle commute with
each element of the finite-memory subgroup $\mathcal{F}_{0}$.
\end{lemma}

\begin{proof}
The proof of this lemma is similar to the proof of Theorem~7 in
Ref.~\cite{arXiv:1105.0649v1}. For completeness, we repeat it here. Consider
that the RHS\ of the last row of (\ref{eq:finite-standard-path}) commutes with
the RHS\ of any of the rows of (\ref{eq:physical-weight-cycle}). This then
implies that the LHS\ of the last row of (\ref{eq:finite-standard-path})
commutes with the LHS\ of any of the rows of (\ref{eq:physical-weight-cycle}).
Since the $S_{i}$ operators already commute, this implies that $M_{1}%
,\ldots,M_{p}$ in (\ref{eq:physical-weight-cycle}) commute with $M_{t-1}$ in
(\ref{eq:finite-standard-path}). This in turn implies that the RHS\ of the
second-to-last row of (\ref{eq:finite-standard-path}) commutes with the
RHS\ of any row of (\ref{eq:physical-weight-cycle}), and furthermore, that the
LHS\ of the second-to-last row of (\ref{eq:finite-standard-path}) commutes
with the LHS\ of any row of (\ref{eq:physical-weight-cycle}). So we have that
$M_{1},\ldots,M_{p}$ in (\ref{eq:physical-weight-cycle}) commutes with
$M_{t-2}$ in (\ref{eq:finite-standard-path}). Continuing this argument
inductively, we can conclude the statement of the lemma.
\end{proof}

Let $C\left(  \mathcal{F}_{0}\right)  $ denote the subgroup consisting of
memory states that commute with each element of the finite-memory subgroup
$\mathcal{F}_{0}$ (so that $C\left(  \mathcal{F}_{0}\right)  $ is the
\textit{centralizer} of $\mathcal{F}_{0}$). The above lemma states that the
zero physical-weight cycle subgroup is a subset of the centralizer of the
finite-memory subgroup:\ $\mathcal{P}_{0}\subseteq C\left(  \mathcal{F}%
_{0}\right)  $. The rest of our proof is dedicated to proving the other
containment:\ $C\left(  \mathcal{F}_{0}\right)  \subseteq\mathcal{P}_{0}$,
implying that the zero physical-weight cycle subgroup $\mathcal{P}_{0}$ is
equivalent to the centralizer of the finite-memory subgroup:\ $\mathcal{P}%
_{0}=C\left(  \mathcal{F}_{0}\right)  $. Having this will then allow us to
easily prove that all non-catastrophic encoders are non-recursive.

Now consider the following transformation corresponding to the response of the
encoder to a Pauli-$Z$ operator on one of the ancilla qubits:%
\begin{equation}%
\begin{tabular}
[c]{c|c|c}%
Mem. & Anc. & Info.\\\hline\hline
$I^{\otimes m}$ & $Z(i)$ & $I^{\otimes k}$%
\end{tabular}
\ \ \ \ \rightarrow\ \
\begin{tabular}
[c]{c|c}%
Phys. & Mem.\\\hline\hline
$h_{i,1}$ & $g_{i,1}$%
\end{tabular}
\ \ \ \ , \label{eq:ancilla-path1}%
\end{equation}
where $i\in\left\{  1,\ldots,n-k\right\}  $. The encoder acts on a finite
number of qubits, so it must either end in the identity memory state or cycle
after following a path from $g_{i,1}$ in which the identity operator is always
input for both the ancilla and information qubits. Thus, we either have
\[%
\begin{tabular}
[c]{c|c|c}%
Mem. & Anc. & Info.\\\hline\hline
$I^{\otimes m}$ & $Z(i)$ & $I^{\otimes k}$\\
$g_{i,1}$ & $I^{\otimes\left(  n-k\right)  }$ & $I^{\otimes k}$\\
$g_{i,2}$ & $I^{\otimes\left(  n-k\right)  }$ & $I^{\otimes k}$\\
$\vdots$ & $\vdots$ & $\vdots$\\
$g_{i,d_{i}-1}$ & $I^{\otimes\left(  n-k\right)  }$ & $I^{\otimes k}$%
\end{tabular}
\ \ \rightarrow%
\begin{tabular}
[c]{c|c}%
Phys. & Mem.\\\hline\hline
$h_{i,1}$ & $g_{i,1}$\\
$h_{i,2}$ & $g_{i,2}$\\
$h_{i,3}$ & $g_{i,3}$\\
$\vdots$ & $\vdots$\\
$h_{i,d_{i}}$ & $I^{\otimes m}$%
\end{tabular}
\ ,
\]
in which case the input $I^{\otimes m}\,|\,Z(i)\,|\,I^{\otimes k}$ leads to a
finite standard path, or we have%
\[%
\begin{tabular}
[c]{c|c|c}%
Mem. & Anc. & Info.\\\hline\hline
$I^{\otimes m}$ & $Z(i)$ & $I^{\otimes k}$\\
$g_{i,1}$ & $I^{\otimes n-k}$ & $I^{\otimes k}$\\
$\vdots$ & $\vdots$ & $\vdots$\\
$g_{i,j-1}$ & $I^{\otimes n-k}$ & $I^{\otimes k}$\\
$\vdots$ & $\vdots$ & $\vdots$\\
$g_{i,t}$ & $I^{\otimes n-k}$ & $I^{\otimes k}$\\
$g_{i,j}$ & $I^{\otimes n-k}$ & $I^{\otimes k}$\\
$\vdots$ & $\vdots$ & $\vdots$\\
$g_{i,t}$ & $I^{\otimes n-k}$ & $I^{\otimes k}$\\
$\vdots$ & $\vdots$ & $\vdots$%
\end{tabular}
\ \ \ \rightarrow\ \
\begin{tabular}
[c]{c|c}%
Phys. & Mem.\\\hline\hline
$h_{i,1}$ & $g_{i,1}$\\
$h_{i,2}$ & $g_{i,2}$\\
$\vdots$ & $\vdots$\\
$h_{i,j}$ & $g_{i,j}$\\
$\vdots$ & $\vdots$\\
$h_{i,t+1}$ & $g_{i,j}$\\
$h_{i,j+1}$ & $g_{i,j+1}$\\
$\vdots$ & $\vdots$\\
$h_{i,t+1}$ & $g_{i,j}$\\
$\vdots$ & $\vdots$%
\end{tabular}
\ ,
\]
so that the encoder enters some cycle. In this latter case, the following
input sequence drives the encoder back to the identity memory state:%
\[%
\begin{tabular}
[c]{c|c|c}%
Mem. & Anc. & Info.\\\hline\hline
$I^{\otimes m}$ & $Z(i)$ & $I^{\otimes k}$\\
$g_{i,1}$ & $I^{\otimes n-k}$ & $I^{\otimes k}$\\
$\vdots$ & $\vdots$ & $\vdots$\\
$g_{i,j-1}$ & $I^{\otimes n-k}$ & $I^{\otimes k}$\\
$g_{i,j}$ & $I^{\otimes n-k}$ & $I^{\otimes k}$\\
$\vdots$ & $\vdots$ & $\vdots$\\
$g_{i,t-j+1}$ & $Z(i)$ & $I^{\otimes k}$\\
$\vdots$ & $\vdots$ & $\vdots$\\
$g_{i,t-1}\cdot g_{i,j-2}$ & $I^{\otimes n-k}$ & $I^{\otimes k}$\\
$g_{i,t}\cdot g_{i,j-1}$ & $I^{\otimes n-k}$ & $I^{\otimes k}$%
\end{tabular}
\ \ \ \rightarrow\ \
\begin{tabular}
[c]{c|c}%
Phys. & Mem.\\\hline\hline
$h_{i,1}$ & $g_{i,1}$\\
$h_{i,2}$ & $g_{i,2}$\\
$\vdots$ & $\vdots$\\
$h_{i,j}$ & $g_{i,j}$\\
$h_{i,j+1}$ & $g_{i,j+1}$\\
$\vdots$ & $\vdots$\\
$h_{i,t-j+2}$ & $g_{i,t-j+2}\cdot g_{i,1}$\\
$\vdots$ & $\vdots$\\
$h_{i,t}$ & $g_{i,t}\cdot g_{i,j-1}$\\
$h_{i,t+1}$ & $g_{i,j}\cdot g_{i,j}=I^{\otimes m}$%
\end{tabular}
\ \ \ .
\]
So the memory states resulting from the input $I^{\otimes m}%
\,|\,Z(i)\,|\,I^{\otimes k}$ are always part of the finite-memory subgroup
$\mathcal{F}_{0}$, and based on Lemma \ref{lemma1}, these memory states
commute with the memory states in any zero-physical-weight cycle.

\begin{lemma}
\label{lemma 4} Suppose that a memory state $M^{\prime}\in C\left(
\mathcal{F}_{0}\right)  $ (that is, it commutes with each element of
$\mathcal{F}_{0}$). Then there is a unique memory state $M\in C\left(
\mathcal{F}_{0}\right)  $, an $S\in\{I,Z\}^{\otimes(n-k)}$, and an
$L\in\left\{  I,X,Y,Z\right\}  ^{\otimes k}$ such that the encoder transforms
$M\ |\ S\ |\ L$ as follows:%
\begin{equation}
M\ |\ S\ |\ L\rightarrow I^{\otimes n}\ |\ M^{\prime}. \label{endpoint}%
\end{equation}

\end{lemma}

\begin{proof}
First, the transition in (\ref{endpoint}) is part of the encoder because we
can always act with the inverse of the encoder on $I^{\otimes n}%
\ |\ M^{\prime}$, producing some unique $M\ |\ S\ |\ L$ as the input
operators. Then, by assumption, $M^{\prime}\in C\left(  \mathcal{F}%
_{0}\right)  $, implying that the RHS of (\ref{endpoint}) commutes with the
RHS of (\ref{eq:ancilla-path1}) for all $i\in\left\{  1,\ldots,n-k\right\}  $
because $g_{i,1}\in\mathcal{F}_{0}$. Thus, the LHS of (\ref{endpoint}) should
commute with the LHS of (\ref{eq:ancilla-path1}) for all $i\in\left\{
1,\ldots,n-k\right\}  $. This implies that $S\in\{I,Z\}^{\otimes(n-k)}$. Also,
$M$ commutes with all states in $\mathcal{F}_{0}$ ($M\in C\left(
\mathcal{F}_{0}\right)  $) for the same reason that the zero physical weight
cycles from Lemma~\ref{lemma1}\ commute with the finite-memory subgroup
$\mathcal{F}_{0}$. This is due to the particular form in (\ref{endpoint}), the
need for input-output commutativity consistency between (\ref{endpoint}) and
(\ref{eq:finite-standard-path}), and the fact that $M^{\prime}\in C\left(
\mathcal{F}_{0}\right)  $.
\end{proof}

\begin{lemma}
\label{lem:start-to-end-edge} Suppose that a memory state $M\in C\left(
\mathcal{F}_{0}\right)  $ (that is, it commutes with each element of
$\mathcal{F}_{0}$). Then there is a unique memory state $M^{\prime}\in
C\left(  \mathcal{F}_{0}\right)  $, an $S\in\{I,Z\}^{\otimes(n-k)}$, and an
$L\in\left\{  I,X,Y,Z\right\}  ^{\otimes k}$ such that the encoder transforms
$M\ |\ S\ |\ L$ as follows:%
\begin{equation}
M\ |\ S\ |\ L\rightarrow I^{\otimes n}\ |\ M^{\prime}.
\label{eq:constructed-edge}%
\end{equation}

\end{lemma}

\begin{proof}
Let $\left\{  T_{1},\ldots,T_{l}\right\}  $ be a generating set for the
finite-memory subgroup $\mathcal{F}_{0}$ and let $\left\{  U_{l+1}%
,\ldots,U_{2m}\right\}  $ be a generating set for the infinite-memory subgroup
$\mathcal{I}_{0}$. For every $i\in\left\{  l+1,\ldots,2m\right\}  $, there
exists a $P_{i}$ and $T_{i}$ such that%
\begin{equation}
U_{i}\ |\ I^{\otimes\left(  n-k\right)  }\ |\ I^{\otimes k}\rightarrow
P_{i}\ |\ T_{i}. \label{eq:inf-mem-trans}%
\end{equation}
Furthermore, the $T_{i}$ operators form a generating set $\left\{
T_{l+1},\ldots,T_{2m}\right\}  $ for $\mathcal{I}_{0}$. If it were not so
(that these $T_{i}$ operators were not independent), we would be able to take
certain multiplicative combinations of the $U_{i}$ operators at the input and
construct a path to the identity memory state, contradicting the fact that the
$U_{i}$ operators form a generating set for elements of the infinite-memory
subgroup $\mathcal{I}_{0}$. Thus, $\left\{  T_{l+1},\ldots,T_{2m}\right\}  $
is a generating set for $\mathcal{I}_{0}$ as well.

Now, we fix $M^{\prime}$ to be the unique Pauli operator on $m$ qubits that
commutes with each element of $\left\{  T_{1},\ldots,T_{l}\right\}  $ and such
that its commutation relation with each $T_{i}$ in $\left\{  T_{l+1}%
,\ldots,T_{2m}\right\}  $ is the same as the commutation relation of $M$ with
the corresponding $U_{i}$ in $\left\{  U_{l+1},\ldots,U_{2m}\right\}  $ (the
correspondence between $U_{i}$ and $T_{i}$ being set by
(\ref{eq:inf-mem-trans})). The uniqueness of $M^{\prime}$ follows from the
fact that specifying the commutation relations of an $m$-qubit Pauli operator
with each element of a generating set for the Pauli group on $m$ qubits
completely specifies that operator up to an irrelevant global phase.
Additionally, the above choice for $M^{\prime}$ implies that $M^{\prime}\in
C\left(  \mathcal{F}_{0}\right)  $.

To obtain the conclusion of the lemma, we consider applying the inverse of the
encoder to $I^{\otimes n}\ |\ M^{\prime}$. By Lemma~\ref{lemma 4}, there
exists an $M^{\prime\prime}\in C\left(  \mathcal{F}_{0}\right)  $, an
$S\in\left\{  I,Z\right\}  ^{\otimes\left(  n-k\right)  }$, and an
$L\in\left\{  I,X,Y,Z\right\}  ^{\otimes k}$ such that%
\begin{equation}
M^{\prime\prime}\ |\ S\ |\ L\rightarrow I^{\otimes n}\ |\ M^{\prime}.
\label{eq:final-remark-lemma9}%
\end{equation}
But this $M^{\prime\prime}$ must be equal to $M$ by construction---the encoder
preserves the commutation relations between the RHSs\ of
(\ref{eq:inf-mem-trans}-\ref{eq:final-remark-lemma9}) and the LHSs of
(\ref{eq:inf-mem-trans}-\ref{eq:final-remark-lemma9}) and we already know that
$M^{\prime\prime},M^{\prime}\in C\left(  \mathcal{F}_{0}\right)  $. As we
stated above, these $2m$\ commutation relations completely specify
$M^{\prime\prime}$ and we can thus conclude that $M^{\prime\prime}=M$. This
concludes the proof of the lemma since we have constructed a unique
$M^{\prime}\in C\left(  \mathcal{F}_{0}\right)  $, $S\in\{I,Z\}^{\otimes
(n-k)}$, and $L\in\left\{  I,X,Y,Z\right\}  ^{\otimes k}$ satisfying
(\ref{eq:constructed-edge}) for all $M\in C\left(  \mathcal{F}_{0}\right)  $.
\end{proof}

\begin{remark}
The above proof works just as well by replacing all the transitions in
(\ref{eq:inf-mem-trans}) with the following ones:%
\[
U_{i}\ |\ S^{\prime} \ | \ I^{\otimes k} \rightarrow P_{i}^{\prime}%
\ |\ T_{i}^{\prime},
\]
for some fixed $S^{\prime}\in\left\{  I,Z\right\}  ^{\otimes\left(
n-k\right)  }$.
\end{remark}

\begin{lemma}
\label{lemma 6} Suppose that a memory state $M\in C\left(  \mathcal{F}%
_{0}\right)  $. Then $M$ is a part of a zero physical-weight cycle, such that
all memory states of this cycle are in $C\left(  \mathcal{F}_{0}\right)  $.
\end{lemma}

\begin{proof}
From Lemma~\ref{lem:start-to-end-edge}, we know that there is a unique memory
state $M_{1}\in C\left(  \mathcal{F}_{0}\right)  $, an $S_{1}\in\left\{
I,Z\right\}  ^{\otimes\left(  n-k\right)  }$, and an $L_{1}\in\left\{
I,X,Y,Z\right\}  ^{\otimes k}$ such that%
\[
M\ |\ S_{1}\ |\ L_{1}\rightarrow I^{\otimes n}\ |\ M_{1},
\]
This is also the case for $M_{1}$. That is, there exists a unique memory state
$M_{2}\in C\left(  \mathcal{F}_{0}\right)  $, an $S_{2}\in\left\{
I,Z\right\}  ^{\otimes\left(  n-k\right)  }$, and an $L_{2}\in\left\{
I,X,Y,Z\right\}  ^{\otimes k}$ such that%
\[
M_{1}\ |\ S_{2}\ |\ L_{2}\rightarrow I^{\otimes n}\ |\ M_{2}.
\]
This process cannot go on indefinitely (that of finding a new $M_{i+1}\in
C\left(  \mathcal{F}_{0}\right)  $ that has not yet already appeared, for a
given $M_{i}\in C\left(  \mathcal{F}_{0}\right)  $). The centralizer $C\left(
\mathcal{F}_{0}\right)  $ of the finite-memory subgroup $\mathcal{F}_{0}$ has
a finite number of elements, and at some point, the transition is guaranteed
to return back to $M$. Also, note that these zero physical-weight transitions
cannot go back to any of $M_{1}$ through $M_{j}$ before going back to
$M$---this would contradict Lemma~\ref{lemma 4}. Thus, $M$ is part of a zero
physical-weight cycle of the following form:%
\begin{equation}%
\begin{tabular}
[c]{c|c|c}%
Mem. & Anc. & Info.\\\hline\hline
$M$ & $S_{1}$ & $L_{1}$\\
$M_{1}$ & $S_{2}$ & $L_{2}$\\
$\vdots$ & $\vdots$ & $\vdots$\\
$M_{j-1}$ & $S_{j-1}$ & $L_{j-1}$\\
$M_{j}$ & $S_{j}$ & $L_{j}$%
\end{tabular}
\ \ \ \ \rightarrow\ \
\begin{tabular}
[c]{c|c}%
Phys. & Mem.\\\hline\hline
$I^{\otimes n}$ & $M_{1}$\\
$I^{\otimes n}$ & $M_{2}$\\
$\vdots$ & $\vdots$\\
$I^{\otimes n}$ & $M_{j}$\\
$I^{\otimes n}$ & $M$%
\end{tabular}
\ \ \ \ .
\end{equation}

\end{proof}

\begin{corollary}
\label{corollary2} From Lemma~\ref{lemma 6}, we conclude that all memory
states of $C\left(  \mathcal{F}_{0}\right)  $ are part of some zero
physical-weight cycle, and when combined with Lemma~\ref{lemma1}, we have that
$\mathcal{P}_{0}=C\left(  \mathcal{F}_{0}\right)  $. So if all memory states
of all zero-physical weight cycles of a transformation commute with a given
memory state, that memory state belongs to $\mathcal{F}_{0}$.
\end{corollary}

\begin{theorem}
[Poulin, Tillich, Ollivier\ \cite{PTO09}]\label{thm:main-theorem}A
non-catastrophic encoder is non-recursive.
\end{theorem}

\begin{proof}
Consider the following transition, resulting from a weight-one logical input:%
\begin{equation}%
\begin{tabular}
[c]{c|c|c}%
Mem. & Anc. & Info.\\\hline\hline
$I^{\otimes m}$ & $I^{\otimes n-k}$ & $X(i)$%
\end{tabular}
\ \ \ \ \rightarrow\ \
\begin{tabular}
[c]{c|c}%
Phys. & Mem.\\\hline\hline
$h$ & $g$%
\end{tabular}
\ \ \ \ . \label{eq:X(i)-path}%
\end{equation}
For a non-catastrophic encoder, all of the edges in a zero physical-weight
cycle have zero logical weight, so that the LHS of (\ref{eq:X(i)-path})
commutes with the LHS of the transitions in these zero physical-weight cycles.
Thus, the RHSs commute as well. Since zero physical-weight cycles have the
identity operator acting on physical qubits, the above memory state $g$
commutes with all memory states of zero-physical weight cycles, so that $g\in
C(\mathcal{P}_{0})$. According to Corollary~\ref{corollary2}, the memory state
$g\in\mathcal{F}_{0}$, and according to Fact~\ref{fact:finite-mem-sub}, there
is a standard path from $g$ to the identity memory state, implying that the
encoder is non-recursive.
\end{proof}

\begin{remark}
Our proof of the above theorem illustrates a particular property of any
non-catastrophic quantum convolutional encoder.\ Transitions of the following
form%
\begin{equation}%
\begin{tabular}
[c]{c|c|c}%
Mem. & Anc. & Info.\\\hline\hline
$I^{\otimes m}$ & $S$ & $L$%
\end{tabular}
\ \ \ \ \rightarrow\ \
\begin{tabular}
[c]{c|c}%
Phys. & Mem.\\\hline\hline
$h_{S,L}$ & $g_{S,L}$%
\end{tabular}
\ \ \ ,
\end{equation}
for all $S\in\{I,Z \}^{\otimes n-k}$ and $L\in\left\{  I,X,Y,Z\right\}
^{\otimes k}\backslash\left\{  I^{\otimes k}\right\}  $ lead to a memory state
$g_{S,L}\in\mathcal{F}_{0}$, from which there is a standard path to the
identity memory state.
\end{remark}

\section{Discussion}

\label{sec:discussion}We have provided a proof of the statement
\textquotedblleft recursive quantum convolutional encoders are
catastrophic\textquotedblright\ that is arguably simpler than the one
originally furnished by Poulin, Tillich, and Olliver in Ref.~\cite{PTO09}. Our
approach was to show that the subgroup of memory states that are part of zero
physical-weight cycles is equivalent to the centralizer of the finite-memory
subgroup. From there, it was straightforward to prove the PTO\ theorem.

The proof given here also provides insight into why entanglement-assisted
quantum convolutional encoders can circumvent this no-go theorem. In
particular, an examination of the example from Figure~4 of Ref.~\cite{WH10b}
reveals that its zero physical-weight cycle subgroup consists of only the
identity memory state, so that its centralizer contains all memory states.
Yet, the finite-memory subgroup for this example is empty. Thus, the symmetry
between these two subgroups is broken by the introduction of entanglement, and
the no-go theorem no longer applies.

\paragraph{Acknowledgements}

We acknowledge Jean-Pierre Tillich for suggesting to us at QIP 2011 in
Singapore that it would be worthwhile to pursue a simpler proof of
Theorem~\ref{thm:main-theorem}. We also acknowledge useful discussions with
Min-Hsiu Hsieh on this topic, and we are grateful 
for the suggestions of the anonymous referee. MMW acknowledges support from the Centre de
Recherches Math\'{e}matiques at the University of Montreal.

\bibliographystyle{plain}
\bibliography{Ref}

\end{document}